\documentclass[smallextended]{svjour3}       
\smartqed  

\usepackage{amssymb,amsmath}
\usepackage{graphicx}
\usepackage{longtable, booktabs}
\usepackage{appendix}
\usepackage{hyperref}
\usepackage{enumerate}

\begin{document}

\title{Optimal strongly conflict-avoiding codes of even length and weight three}

\titlerunning{Optimal SCAC of even length and weight three}        

\author{Yijin Zhang \and Yuan-Hsun Lo \and \\ Wing Shing Wong}

\authorrunning{Y. Zhang et al.} 

\institute{Y. Zhang \at
              School of Electronic and Optical Engineering, Nanjing University of Science and Technology, Nanjing, China \\
              \email{yijin.zhang@gmail.com}
           \and
           Y.-H. Lo \at
              Department of Mathematics, National Taiwan Normal University, Taipei 116, Taiwan\\
              \email{yhlo0830@gmail.com}
           \and
           W. S. Wong \at
              Department of Information Engineering, the Chinese University of Hong Kong, Hong Kong\\
              \email{wswong@ie.cuhk.edu.hk}
}

\date{Received: date / Accepted: date} 


\renewcommand{\labelenumi}{(\roman{enumi})}

\newcommand{\seq}[1]{\mathbf{#1}}

\newcommand{\set}[1]{\mathcal{#1}}
\newcommand{\CAC}{\mathsf{CAC}}
\newcommand{\SCAC}{\mathsf{SCAC}}
\newcommand{\C}{\mathcal{C}}

\maketitle

\begin{abstract}
Strongly conflict-avoiding codes (SCACs) are employed in a slot-asynchronous multiple-access collision channel without feedback to guarantee that each active user can send at least one packet successfully in the worst case within a fixed period of time.
By the assumption that all users are assigned distinct codewords, the number of codewords in an SCAC is equal to the number of potential users that can be supported.
SCACs have different combinatorial structure compared with conflict-avoiding codes (CACs) due to additional collisions incurred by partially overlapped transmissions.
In this paper, we establish upper bounds on the size of SCACs of even length and weight three.
Furthermore, it is shown that some optimal CACs can be used to construct optimal SCACs of weight three.
\end{abstract}

\section{Introduction}
\subsection{Motivation}
The collision channel without feedback model \cite{Massey85} is investigated in this paper.
There are total $M$ potential users and at most $k$ users are active at the same time.
{\em Protocol sequences}~\cite{GyorfiVajda93,Nguyen92,SCSW,CRT,Wong07,CIS} are used to provide multiple-access.
Let $x_i = (x_{i,0}, x_{i,1}, \ldots , x_{i,L-1})$ be a binary protocol sequence with length $L$ assigned to user $i$.
Each active user sends its packet to a common sink if and only if the assigned sequence value equals one.
The channel time is partitioned into fixed-length slots and the packet length exactly occupies a slot.
A total overlap of packets occurs if more than one user start their transmission simultaneously; and a partial overlap of packets occurs if one packet starts or ends its transmission within the transmission duration of some other packet.
Any partial or total overlap of packets would incur a collision.
A packet without suffering from any collision is received error-free; otherwise it is assumed to be unrecoverable.
As there is no feedback from the receiver and no cooperation among the users, each user has a relative delay offset.
Let $\delta_i$ be the time offset of user $i$ for $i=1,2,\ldots, M$, measured in time slot duration units.
As introduced in~\cite{Massey85}, there are two different levels of synchronization:
\begin{enumerate}
  \item The channel is {\em slot-synchronized} if all users start transmitting at the slot boundaries, i.e., the time offsets $\delta_1, \delta_2,\ldots,  \delta_M$ are arbitrary integers.
        Collisions will result only when packets totally overlap.
  \item The channel is {\em slot-asynchronous} if all users do not know the slot boundaries of the channel, i.e., the time offsets $\delta_1, \delta_2,\ldots,  \delta_M$ are arbitrary real numbers.
    Some collisions may be incurred by partial overlap of packets.
\end{enumerate}

A set of $M$ binary sequences $\{x_1,x_2,\ldots,x_M\}$ is said to be an {\em $(M, k, \omega, L, \sigma)$ protocol sequence set}~\cite{Nguyen92} if any sequence is of length $L$, Hamming weight $\omega$, and has the property that each active user can transmit at least $\sigma$ packets successfully in a period of $L$ slots in the worst case.
When $\sigma \geq 1$, we say this sequence set enjoys the \emph{nonblocking property}.
Obviously, whether $\sigma \geq 1$ or not highly depends on the assumption of synchronization.

Let $\set{I}$ be a codeword of weight~$\omega$ over $\mathbb{Z}_L$.
Since a binary sequence of length $L$ can be identified with a subset of $\mathbb{Z}_L$ representing the indices of nonzero positions, a set of $M$ protocol sequences can be viewed as a code consisting of $M$ codewords.
In order to provide the nonblocking property at different levels of synchronization, the following two classes of codes have been studied as protocol sequences extensively in the literature.
\begin{enumerate}
  \item  An $(M, k, \omega, L, \sigma)$ protocol sequence set is a {\em conflict-avoiding code} (CAC)~\cite{Fu,Fu_Lo_Shum_2012,JMJTT,Levenshtein07,LT05,Momihara07b,Momihara07,Mishima09,ShumWong09,CAC10} if $k=\omega$ and $\sigma = 1$ in the slot-synchronized case.
  \item An $(M, k, \omega, L, \sigma)$ protocol sequence set is a {\em strongly conflict-avoiding code} (SCAC)~\cite{SCAC} if $k=\omega$ and $\sigma = 1$ in the slot-asynchronous case.
      SCACs consider a more practical channel model.
\end{enumerate}
As $k=w$, both CACs and SCACs require that there is at most one collision between any two distinct sequences for any relative delay offsets.
However, collisions incurred by partially overlapped transmissions need to be additionally considered in the design of SCACs.
This yields different combinatorial structures of CACs and SCACs, as argued in~\cite{SCAC}.
Before presenting them accordingly in Section~\ref{sec:preliminary}, we provide an example first as the following.
\begin{align*}
x_1 &= (1,1,1,0,0,0,0,0,0,0,0,0) \\
x_2 &= (1,0,0,1,0,0,1,0,0,0,0,0) \\
x_3 &= (1,0,0,0,1,0,0,0,1,0,0,0)
\end{align*}
$\{x_1,x_2,x_3\}$ forms a CAC with $M=3$ and $L=12$.
However, it is not an SCAC.
For $\delta_1=1$, $\delta_2=1.5$ and $\delta_3=3$, all packets from user 1 are lost due to two partial overlappings and one complete overlapping, as illustrated in Fig.~\ref{fig:packet}.

\begin{figure}
\begin{center}
  \includegraphics[width=3.2in]{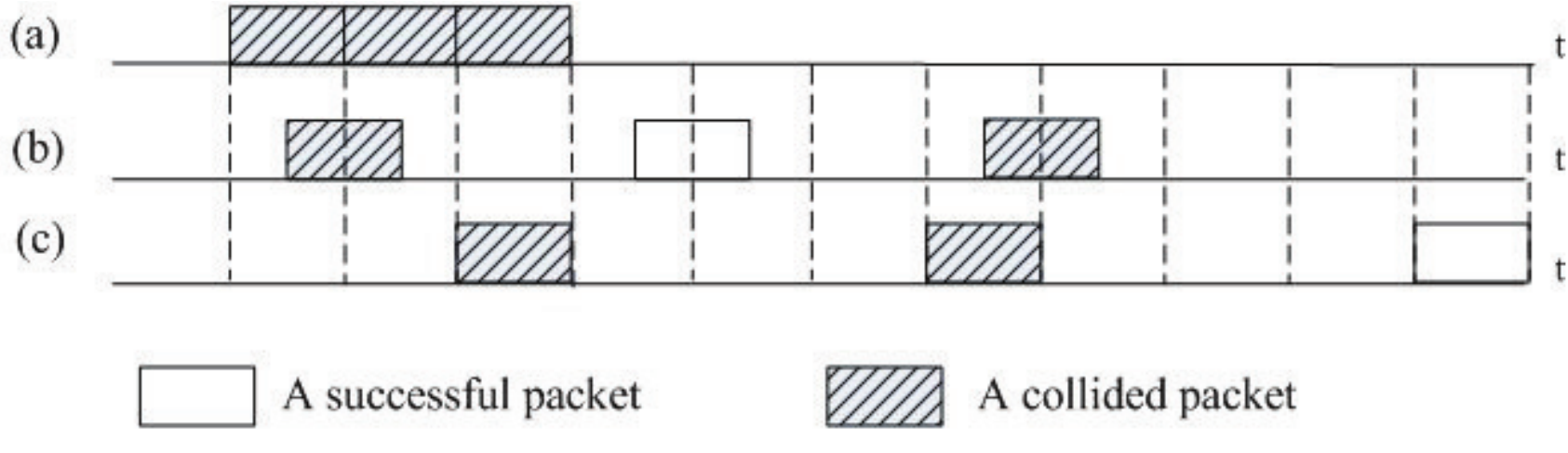}
\end{center}
\caption{(a) Packets from user 1, (b) packets from user 2, (c) packets from user 3.}
\label{fig:packet}
\end{figure}

In the study of CAC or SCAC, the main theme is to find as many sequences (or, codewords) as possible, for a given pair of integers $L$ and $w$.
If a CAC or SCAC enjoys the maximal size of codewords, then this code is said to be {\em optimal}.

Asymptotically optimal and optimal CACs for general weights were investigated in~\cite{ShumWong09,CAC10}.
Based on previously known constructions of CACs, asymptotically optimal SCACs are derived in \cite{SCAC} under the assumption that each codeword possesses a special structure, called equi-difference.
Moreover, optimal CACs of weight three are investigated in \cite{Fu,Fu_Lo_Shum_2012,JMJTT,Levenshtein07,LT05,Momihara07b,Mishima09}.
The code size spectrum of optimal CACs with even length and weight three has been completely settled by these studies. 
However, relatively little is known about the code size of optimal SCACs.
In this paper, we are going to find optimal SCACs of even length and weight three, which can be applied to more realistic scenarios.

The rest of this paper is organized as follows.
In Section~\ref{sec:preliminary}, we introduce some relevant definitions and relative known results in the literatures, as well as present a necessary condition for the existence of an SCAC.
Several useful properties of codewords in an SCAC are given in Section~\ref{sec:property}.
New upper bounds on the size of SCACs are derived in Section~\ref{sec:bound}.
In Section~\ref{sec:optimal} we prove that some upper bounds in Section~\ref{sec:bound} are indeed tight in several cases.
Finally, conclusions are given in Section~\ref{sec:conclusion}.

\section{Preliminaries} \label{sec:preliminary}
\subsection{Definitions and notations}

Let $\mathbb{Z}_L = \{0,1,\ldots,L-1\}$ denote the ring of residues modulo $L$ and $\mathcal{P}(L,\omega)$ denote the set of all $\omega$-subsets of $\mathbb{Z}_L$.
Each element $x \in \mathcal{P}(L,\omega)$ can be identified with a binary sequence of length $L$ and weight $\omega$ representing the indices of the nonzero positions.
Therefore, a CAC or SCAC of length $L$ and weight $\omega$ can be viewed as a subset of $\mathcal{P}(L,\omega)$.
We call elements in $\mathcal{P}(L,\omega)$ \emph{codewords}.

For a codeword $\set{I}\in\mathcal{P}(L,\omega)$, let $d(\set{I}) :=\{a-b \text{ (mod } L):\, a,b\in \set{I}\}$
denote the \emph{set of differences} between pairs of elements in $\set{I}$, and let $d^*(\set{I}):=d(\set{I}) \setminus \{0\}$
denote the \emph{set of non-zero differences} in $\set{I}$.
Then a formal definition of a CAC can be given as follows.

\begin{definition}
A CAC of length $L$ and weight $\omega$ is a subset $\mathcal{C} = \{\set{I}_1,...,\set{I}_M\} \subset \mathcal{P}(L,\omega)$ satisfying the condition that for all $j\neq k$,
\begin{equation}
 d^*(\set{I}_j) \cap d^*(\set{I}_k) = \emptyset. \label{eq:CAC}
\end{equation}
\end{definition}

For given $L$ and $w$, let $\CAC(L,\omega)$ denote the class of all CACs of length $L$ and weight $w$.
The maximum size of a code in $\CAC(L,\omega)$ is denoted by $M(L,\omega)$.
A code $\set{C}\in\CAC(L,\omega)$ is said to be \emph{optimal} if $|\set{C}|=M(L,\omega)$.

\smallskip
Given two subsets $\set{A},\set{B}\subset\mathbb{Z}_L$, let $\set{A}\pm \set{B} := \{a\pm b \text{ (mod }L):\,a\in\set{A},b\in\set{B}\}$.
Then an SCAC can also be defined by means of $d$ and $d^*$.

\begin{definition}
An SCAC of length $L$ and weight $\omega$ is a subset $\mathcal{C} = \{\set{I}_1,...,\set{I}_M\} \subset \mathcal{P}(L,\omega)$ satisfying the condition that for all $j\neq k$,
\begin{equation}
\Big(d^*(\set{I}_j) \cup (d^*(\set{I}_j)+\{1\})\cup (d^*(\set{I}_j)-\{1\})\Big) \cap d(\set{I}_k) = \emptyset. \label{eq:SCAC}
\end{equation}
This definition captures all the possibilities of partial collisions in slot asynchronous systems.
\end{definition}

Similarly, for given $L$ and $w$, let $\SCAC(L,\omega)$ denote the class of all SCACs of length $L$ and weight $w$.
The maximum size of a code in $\SCAC(L,\omega)$ is denoted by $M_S(L,\omega)$.
A code $\set{C}\in\SCAC(L,\omega)$ is said to be \emph{optimal} if $|\set{C}|=M_S(L,\omega)$.

Given a code $\set{C}$ in $\CAC(L,\omega)$ or $\SCAC(L,\omega)$, a codeword $\set{I}\in\set{C}$ is called \emph{equi-difference} if all its elements form an arithmetic progression in $\mathbb{Z}_L$, i.e., $\set{I} = \{0, g, 2g, \ldots, (\omega-1)g\}$ for some $g\in\mathbb{Z}_L$, where the product $jg$ is calculated modulo $L$.
The element $g$ is called the \emph{generator} of $\set{I}$.
Without loss of generalization, we assume $g\leq L/2$ in this paper.
A code is called \emph{equi-difference} if it entirely consists of equi-￼difference codewords.
We use $M^e(L,\omega)$ (or $M_S^e(L,\omega)$) to denote the maximum code size among all equi-difference CACs (or SCACs) of length $L$ and weight $w$.

For a codeword $\set{I}\in\mathcal{P}(L,\omega)$ define the \emph{set of shifted non-zero differences} of $\set{I}$ by $d^+(\set{I}) := d^*(\set{I}) + \{0,1\}$.
Then the definition of an SCAC can be rewritten as follows.

\begin{proposition}[\cite{SCAC}]
$\C=\{\set{I}_1,\set{I}_2,\ldots, \set{I}_M\}\in\SCAC(L,\omega)$ if and only if
\begin{enumerate}
  \item $\{1,L-1\}\cap d^*(\set{I}_j)=\emptyset$ for all $j$; and
  \item $ d^+(\set{I}_j)\cap d^+(\set{I}_k) = \emptyset$ for all $j \neq k$.
\end{enumerate}
\label{prop:property}
\end{proposition}

Proposition~\ref{prop:property} implies directly that for any $\set{C}\in\SCAC(L,\omega)$, the following holds:
\begin{equation}\label{eq:d+_union}
\bigcup_{\set{I}\in\set{C}} d^+(\set{I}) \subseteq \{2,3,\ldots,L-1\}.
\end{equation}

Let $\set{A}$ be a subset of $\mathbb{Z}_L$.
A subset of $\set{A}$ which consists of consecutive integers is called a closed interval.
A closed interval $S$ is maximal if for any other closed interval $T$, either $T\subset S$ or $T\cap S=\emptyset$.
Obviously, $\set{A}$ can be uniquely partitioned into several maximal closed intervals, called \emph{tubes}.
A tube is denoted by $T(x,y)$ if its smallest and largest integer are $x$ and $y$, respectively.
$T(x,y)$ is called \emph{O-rough} if $x$ and $y$ are both odd, \emph{E-rough} if $x$ and $y$ are both even, and \emph{flat} otherwise.

On the other hand, $\{2,3,\ldots,L-1\} \setminus \set{A}$ can also be uniquely partitioned into several maximal closed intervals.
They can be viewed as \emph{gaps} in $\set{A}$.
Note that the elements $0,1$ are not taken into consideration because in what follows, we will focus on $\set{A}$'s which are shifted non-zero difference set of some codeword in an SCAC and thus $0,1\notin\set{A}$ by \eqref{eq:d+_union}.
We denote a gap with the smallest integer $x$ and largest integer $y$ by $G(x,y)$.
Similar to tubes, we also classify gaps into E-rough, O-rough and flat gaps.
Note that it is possible $x=y$ for some gaps but not for tubes of shifted non-zero difference sets.

Assume that $\set{C}$ is an SCAC and $\set{I}_j$ is one of its codewords.
We use $T_j(x,y)$ (resp., $G_j(x,y)$) to emphasize a tube (resp., a gap) in the shifted non-zero difference set $d^+(\set{I}_j)$.
For example, let $\set{I}_1 = \{0,4,7\}$ be one codeword in some code $\set{C}\in\SCAC(26,3)$. Then $d^+(\set{I}_1)=\{3,4,5,7,8,19,20,22,23,24\}$. 
There are one O-rough tube $T_1(3,5)$; two flat tubes $T_1(7,8)$, $T_1(19,20)$; and one E-rough tube $T_1(22,24)$.
On the other hand, there are two E-rough gaps $G_1(2,2), G_1(6,6)$; one flat gap $G_1(9,18)$; and two O-rough gaps $G_1(21,21)$, $G_1(25,25)$.

Now, we define a special gap, called \emph{solitary} gap, in a code.

\begin{definition}
Consider a given SCAC, $\set{C}$, and one of its codewords $\set{I}_j$.
Let $G_j(x,y)$ be a gap in $d^+(\set{I}_j)$ and $T(x',y')$ be a tube in $\bigcup_{\set{I} \in \C} d^+(\set{I})$.
If $x\leq x'$ and $y'\leq y$, then this tube is said to be included in the gap, denoted by $T(x',y') \vartriangleleft G_j(x,y)$.
An E-rough (or O-rough) gap $G_j(x,y)$ is said to be \emph{solitary} if there is no E-rough (or O-rough) tube $T(x',y')$ in $\bigcup_{\set{I} \in \C} d^+(\set{I})$ such that  $T(x',y') \vartriangleleft G_j(x,y)$.
\end{definition}

For example, let $\set{I}_1=\{0,2,4\}$, $\set{I}_2=\{0,6,12\}$ and $\set{I}_3=\{0,9,19\}$ be the three codewords in a code $\set{C}\in\SCAC(28,3)$.
Then it can be checked that $G_3(2,8)$ is solitary.

\subsection{Previously known results} \label{sec:pre_known}
We summarize some previously known deterministic results on CACs and SCACs of weight three in this subsection.

\begin{theorem}[\cite{LT05}]
$M(L,3)=M^e(L,3)=(L-2)/4$ for any $L\equiv 2$ (mod 4).
\label{thm:LT05}
\end{theorem}

\begin{theorem}[\cite{Fu,JMJTT,Mishima09}]
Let $L=4t$. Then
\[
M(L,3) = 
\begin{cases}
    7L/64 & \text{if } t\equiv 0 \text{ (mod $8$)},\\
    (7L+8)/64 & \text{if } t\equiv 1 \text{ (mod $8$)},\\
    (7L-48)/64 & \text{if } t\equiv 2,10 \text{ (mod $24$)},\\
    (7L+24)/64 & \text{if } t\equiv 3 \text{ (mod $24$)},\\
    (7L-32)/64 & \text{if } t\equiv 4,20 \text{ (mod $24$)},\\
    (7L-24)/64 & \text{if } t\equiv 5,13 \text{ (mod $24$)},\\
    (7L-16)/64 & \text{if } t\equiv 6 \text{ (mod $8$)},\\
    (7L-8)/64 & \text{if } t\equiv 7 \text{ (mod $8$)},\\
    (7L-40)/64 & \text{if } t\equiv 11,19 \text{ (mod $24$)},\\
    (7L+32)/64 & \text{if } t\equiv 12 \text{ (mod $24$)},\\
    (7L+16)/64 & \text{if } t\equiv 18 \text{ (mod $24$)},\\
    (7L+40)/64 & \text{if } t\equiv 21 \text{ (mod $24$)}.\\
\end{cases}
\]
\label{thm:CAC_4t}
\end{theorem}

\begin{theorem}[\cite{Wu}]
The followings hold.
\begin{enumerate}
\item $M(L,3)=M^e(L,3)=(L-1)/4$ if $L=2^{2t}+1$ for $t\geq 1$.
\item $M(L,3)=M^e(L,3)=(L+1)/4$ if $L=2^{2^t}-1$ for $t\geq 2$.
\end{enumerate}
\label{thm:knowCAC}
\end{theorem}

\begin{theorem}[\cite{Mishima14}]
$M(L,3)=M^e(L,3)=(L-1)/4$ if
\begin{enumerate}
\item $L=2^{2t-1}-2^{t}+1$ for $t\geq 2$, or
\item $L=2^{2t-1}+2^{t}+1$ for $t\geq 1$.
\end{enumerate}
\label{thm:knowCAC2}
\end{theorem}

As for SCACs of weight three, there are few results reported in the literature.
An exception is the following.

\begin{theorem}[\cite{SCAC}]
Let $L$ be an integer factorized as $3^q 7^r \ell$, where $\ell$ is an even integer not divisible by 3 or 7. Then for $L\geq 18$ we have
\[
 M_S(L,3) \leq \begin{cases}
 \big\lfloor (L-2)/6 \big\rfloor & \text{ if }  q=r=0, \\
 \big\lfloor L/6 \big\rfloor & \text{ if }  q\geq1,r=0,\\
 \big\lfloor (L-1)/6 \big\rfloor & \text{ if }  q=0,r\geq1,\\
 \big\lfloor (L+1)/6 \big\rfloor & \text{ if }  q\geq1,r\geq1.\\
 \end{cases}
\]
\label{thm:knowSCAC}
\end{theorem}

\subsection{A Necessary Condition}
We close this section with the following necessary condition for the existence of an SCAC.
The result delineates the impact of solitary gaps and is based on SCAC characteristics presented in Proposition~\ref{prop:property}.

\begin{lemma}
Consider a given code $\set{C}\in\SCAC(L,\omega)$. 
If there exists one codeword, say $\set{I}_j$, having  $\lambda$ solitary gaps in $d^+(\set{I}_j)$, then 
\[
L \geq 2+\lambda+\sum_{\set{I}\in \C} |d^+(\set{I})|.
\]
\label{lemma:gap}
\end{lemma}
\begin{proof}
Let $G_j(x,y)$ be one of the $\lambda$ solitary gaps in $d^+(\set{I}_j)$.
We assume that $G_j(x,y)$ is E-rough, i.e., $x$ and $y$ are both even.
This implies the number of even integers in $G_j(x,y)$ is one more than that of odd integers.
By the definition of the solitary gap, we cannot find an E-rough tube in $\bigcup_{\set{I} \in \C} d^+(\set{I})$, say $T(x',y')$, such that $T(x',y') \vartriangleleft G_j(x,y)$.
From the defining property of flat and O-rough tubes, we know the number of odd integers in a flat or O-rough tube is equal to or bigger than that of even integers.
Thus we always can find an even integer in $G_j(x,y)$ which is not included in $\bigcup_{\set{I} \in \C} d^+(\set{I})$.
For the case $G_j(x,y)$ is O-rough, the proof goes along the same line as above and is omitted.
The result is that there exists an odd integer not included in $\bigcup_{\set{I} \in \C} d^+(\set{I})$.

We conclude that at least $\lambda$ integers in the interval $[2,L-1]$ do not belong to $\bigcup_{\set{I} \in \C} d^+(\set{I})$, since there exist $\lambda$ solitary gaps in $d^+(\set{I}_j)$.
Following Proposition~\ref{prop:property}, we finally obtain that
\[
L-2-\lambda \geq |\bigcup_{\set{I} \in \C} d^+(\set{I})|=\sum_{\set{I} \in \C} |d^+(\set{I})|.
\]
\qed
\end{proof}

\section{Property of Codewords} \label{sec:property}
Lemma~\ref{lemma:gap} provides a recipe for upper bounding the size of SCAC, which relies on $|d^+(\set{I})|$ for different codewords.
In this section, we derive $|d^+(\set{I})|$ for any codeword $\set{I}$.
The following definition is useful for the evaluation of $|d^+(\set{I})|$.

\begin{definition}
We adopt the terminology in~\cite{SCAC} and say that a codeword $\set{I}$ is {\em dispersive} if any two distinct elements in $d(\set{I})$ are not consecutive. Otherwise, it is {\em non-dispersive}.
\end{definition}

By Proposition~\ref{prop:property}(i), $|d^+(\set{I})|=2|d^*(\set{I})|$ if $\set{I}$ is a dispersive codeword in an SCAC.

\subsection{Non-equi-difference Codewords}

Let $\mathcal{I}=\{0,q_1,q_1+q_2\}$ be a non-equi-difference codeword in a code $\set{C}\in\SCAC(L,3)$ for some $q_1,q_2 \geq 2$ and $q_1+q_2 < L$.
After setting $q_3=L-q_1-q_2$, we have
\[
d^*(\set{I})=\{q_1,q_2,q_3,L-q_1,L-q_2,L-q_3\}.
\]
Now, we write $q_1,q_2,q_3$ in an ascending order as $q_l,q_m,q_u$.
Since $\mathcal{I}$ is non-equi-difference, $q_l,q_m,q_u$ must be mutually distinct and thus
\begin{equation}
q_l<q_m<q_u,L-q_u<L-q_m<L-q_l.  \label{eq:f}
\end{equation}
Therefore, 
\begin{equation} \label{eq:d_star_nonequal}
 |d^*(\set{I})|  = \begin{cases}
 5  & \text{ if } \ q_u=L/2, \\
 6  & \text{ if } \ q_u\neq L/2.
 \end{cases}
\end{equation}

\begin{lemma}
Let $\set{I}$ be a non-equi-difference codeword in a code $\set{C}\in\SCAC(L,3)$ with even $L$ and $d^*(\set{I})=\{q_l,q_m,q_u,L-q_u,L-q_m,L-q_l\}$, where the three parameters $q_l, q_m, q_u$ satisfy $q_l+q_m+q_u=L$ and the inequality in \eqref{eq:f}.
If $q_u<L/2$, then 
\begin{enumerate}
 \item $|d^+(\set{I})| = 8$ if $q_l + 1 = q_m =q_u-1=L/3$; and
 \item $|d^+(\set{I})| \geq 10$ otherwise.
\end{enumerate}
 \label{lemma:nonc1}
\end{lemma}
\begin{proof}
By the assumption that $q_u<L/2$, \eqref{eq:f} can be written as
\begin{equation}
q_l < q_m < q_u < L-q_u < L-q_m <L-q_l,
\label{eq:f1}
\end{equation}
and thus $|d^*(\set{I})|=6$.
Moreover, $q_l\geq 2$ and $q_u<L/2$ imply respectively that $L-q_l+1<L$ and $q_u+1<L-q_u$.
Then we have
\[
d^+(\set{I}) \supseteq d^*(\set{I}) \uplus \{q_u+1,L-q_l+1\}.
\]
Note that the notation $\uplus$ refers to disjoint union operation, which is used to emphasize that the two involved sets are disjoint.

If $q_l+1=q_m=q_u-1$, then $q_m=L/3$, and $d^+(\set{I})$ is exactly equal to $d^*(\set{I}) \uplus \{q_u+1,L-q_l+1\}$. 
Hence $|d^+(\set{I})|=8$ in this case.

If $q_l+1\neq q_m$, then $q_l+1$ and $L-q_m+1$ are included in $d^+(\set{I})$ but not $d^*(\set{I})$.
Similarly, if $q_m\neq q_u-1$, then $q_m+1$ and $L-q_u+1$ are in $d^+(\set{I})\setminus d^*(\set{I})$.
In either case, we obtain $|d^+(\set{I})|\geq 10$.
This completes the proof.
\qed
\end{proof}

For example, let $L=24$.
If $\set{I}=\{0,8,15\}$, then $d^*(\set{I})=\{7,8,9,15,16,17\}$ and $|d^+(\set{I})|=8$.
If $\set{I}=\{0,6,13\}$, then $d^*(\set{I})=\{6,7,11,13,17,18\}$ and $|d^+(\set{I})|=10$.

\begin{lemma}
Let $\set{I}$ be a non-equi-difference codeword in a code $\set{C}\in\SCAC(L,3)$ with even $L$ and $d^*(\set{I})=\{q_l,q_m,q_u,L-q_u,L-q_m,L-q_l\}$, where the three parameters $q_l, q_m, q_u$ satisfy $q_l+q_m+q_u=L$ and the inequality in \eqref{eq:f}.
If $q_u\geq L/2$, then
\begin{enumerate}
  \item $ |d^+(\set{I})| = 8 $ if $q_m=q_l+1=(L+2)/4$, $q_u=L/2$; and
  \item $|d^+(\set{I})| \geq 10$ otherwise.
\end{enumerate}
\label{lemma:nonc2}
\end{lemma}
\begin{proof}
We first consider $q_u>L/2$.
In this case, \eqref{eq:f} can be written as
\[
q_l < q_m < L-q_u < q_u < L-q_m <L-q_l.
\]
It is easy to see that $L-q_u+1$ and $L-q_l+1$ are in $d^+(\set{I})\setminus d^*(\set{I})$.
We now claim that $q_m+1$ and $q_u+1$ are also in $d^+(\set{I})\setminus d^*(\set{I})$.
Suppose the assertion is not true; that is, $q_m+1=L-q_u$.
By the assumption that $q_l+q_m+q_u=L$, we have $q_l=1$, which contradicts to Proposition~\ref{prop:property}(i).
Therefore, 
\[
d^+(\set{I}) \supseteq d^*(\set{I}) \uplus \{q_m+1,L-q_u+1,q_u+1,L-q_l+1\},
\]
and thus $|d^+(\set{I})|\geq 10$.

As for the case of $q_u=L/2$, \eqref{eq:f} can be written as
\[
q_l < q_m < q_u = L-q_u < L-q_m <L-q_l.
\]
By the same argument, $q_m+1,q_u+1$ and $L-q_l+1$ are in $d^+(\set{I})\setminus d^*(\set{I})$.
Then we have
\[
d^+(\set{I}) \supseteq d^*(\set{I}) \uplus \{q_m+1,q_u+1,L-q_l+1\}.
\]
If $q_l+1=q_m$, then $q_m$ must be equal to $(L+2)/4$ and $|d^+(\set{I})|=8$.
If $q_l+1< q_m$, then $q_l+1$ and $L-q_m+1$ will be in $d^+(\set{I})\setminus d^*(\set{I})$, and thus $|d^+(\set{I})|\geq 10$.
\qed

%
%
%
\end{proof}

For example, let $L=26$. 
If $\set{I}=\{0,6,13\}$, then $d^*(\set{I})=\{6,7,13,19,20\}$ and $|d^+(\set{I})|=8$.
If $\set{I}=\{0,5,13\}$, then $d^*(\set{I})=\{5,8,13,18,21\}$ and $|d^+(\set{I})|=10$.

\begin{proposition}
Let $\set{I}$ be a non-equi-difference codeword in a code $\set{C}\in\SCAC(L,3)$ with even $L$ such that $|d^+(\set{I})| < 10$ and has at least one rough tube.
Assume that $d^*(\set{I})=\{q_l,q_m,q_u,L-q_u,L-q_m,L-q_l\}$, where the three parameters $q_l, q_m, q_u$ satisfy $q_l+q_m+q_u=L$ and the inequality in \eqref{eq:f}.
Then, $q_m=q_l+1=(L+2)/4$, $q_u=L/2$.
\label{prop:tube}
\end{proposition}
\begin{proof}
By Lemma~\ref{lemma:nonc1} and Lemma~\ref{lemma:nonc2}, there are two possible codewords satisfying $|d^+(\set{I})| < 10$.
They are the codeword with $q_m=q_l+1=q_u-1=L/3$ and that with $q_m=q_l+1=(L+2)/4, q_u=L/2$, and both have $|d^+(\set{I})| = 8$.
By the definition of rough tubes, only the latter one has at least one rough tube.
\qed
\end{proof}

\subsection{Equi-difference Codewords}
We start this subsection with the following known result on equi-difference codewords.
\begin{lemma}[\cite{LT05}]
Let $\set{C}\in\CAC(L,3)$ and $\set{I}$ be one of its equi-difference codewords.
Then we have
\[
 |d^*(\set{I})|  = \begin{cases}
 2  & \text{ if } \ g=L/3, \\
 3  & \text{ if } \ g=L/4, \\
 4  & \text{ otherwise }.
 \end{cases}
\]
\label{lemma:g1}
\end{lemma}

Lemma~\ref{lemma:g1} obviously holds for the case of $\set{C}\in\SCAC(L,3)$ due to $\SCAC(L,3)\subseteq\CAC(L,3)$.
A codeword $\set{I}$ in a CAC or SCAC of weight three is called {\em exceptional} \cite{Momihara07} if $ |d^*(\set{I})| < 4$.
Therefore, there are at most two exceptional equi-difference codewords in a CAC or SCAC of weight three.

\begin{lemma}[{\cite{SCAC}}]
Let $\set{I}$ be a non-dispersive equi-difference codeword with generator $g$ in an code in $\SCAC(L,\omega)$. If there are $k$ ($k>0$) pairs of consecutive elements in $d^*(\set{I})$, then we have
\begin{enumerate}
\item $ (2w-k-1)g \equiv \pm 1$ (mod $L$) with $k\leq w-1$;
  \item $g$ and $2w-k-1$ are both relatively prime to $L$;
  \item $\set{I}$ is non-exceptional.
\end{enumerate}
\label{thm:equi3}
\end{lemma}

Following Lemma~\ref{thm:equi3} we have:
\begin{corollary}
Let $\set{I}$ be a non-dispersive equi-difference codeword with generator $g$ in a code in $\SCAC(L,3)$ with even $L$.
Then there are two pairs of consecutive elements in $d^*(\set{I})$ and
\begin{equation}
 g= \frac{L+1}{3} \ or \ \frac{L-1}{3}. \label{eq:g2}
\end{equation}
\label{lemma:g2}
\end{corollary}
\begin{proof}
Suppose there are $k$ ($\geq 1$) pairs of consecutive elements in $d^*(\set{I})$.
Since $L$ is even and $\omega=3$, by Lemma~\ref{thm:equi3}(i)--(ii), we have $k=2$, $\gcd(g,L)=1$ and
\[
 3g\equiv \pm 1 \text{ (mod L)}.
\]
By the assumption that $g\leq L/2$, we have $3g<2L$, and thus the above equation can be reduced to
\[
 g= (L+1)/3 \ \ or \ \ (L-1)/3.
\]
Note that $\set{I}$ is non-exceptional by Lemma~\ref{thm:equi3} (iii).
\qed
\end{proof}

Now we are ready to derive results on $|d^+(\set{I})|$ for a different type of equi-difference $\set{I}$ as follows.

\begin{theorem}
Let $\set{I}$ be an equi-difference codeword with generator $g$ in a code in $\SCAC(L,3)$ with even $L$.
Then we have
\[
 |d^+(\set{I})|  = \begin{cases}
 4  & \text{ if } \ g=\frac{L}{3}, \\
 6  & \text{ if } \ g=\frac{L}{4} \ or \ \frac{L+1}{3} \ or \ \frac{L-1}{3}, \\
 8  & \text{ otherwise }.
 \end{cases}
\]
\label{thm:cs}
\end{theorem}
\begin{proof}
Corollary~\ref{lemma:g2} promises that there is only one non-dispersive equi-difference codeword: $g=(L+1)/3$ or $(L-1)/3$. 
In either case, we always have $|d^+(\set{I})| = 6$.

We now consider that $\set{I}$ is dispersive.
It is obvious that $|d^+(\set{I})|=2|d^*(\set{I})|$.
Then the result follows from Lemma~\ref{lemma:g1}.

%
%
%
\qed
\end{proof}

As proved in Lemma~\ref{lemma:nonc1},~Lemma~\ref{lemma:nonc2} and Theorem~\ref{thm:cs}, we conclude that in an SCAC with even length and weight three there are four types of codeword $\set{I}$ satisfying $|d^+(\set{I})|<8$, each of which is equi-difference.
We classify them in Table~\ref{table:w8} with notations $E_1,E_2,N_1,N_2$, and make an illustration by the following example.

\begin{table}[h]
\[
\begin{array}{|c|c|c|} \hline
\text{Codeword} & \text{Generator} & |d^+(\set{I})|  \\ \hline \hline
\set{I}_{E_1} & L/4  & 6   \\ \hline
\set{I}_{E_2} & L/3  & 4   \\   \hline
\set{I}_{N_1} & (L-1)/3  & 6   \\   \hline
\set{I}_{N_2} & (L+1)/3  & 6   \\ \hline
\end{array}
\]
\caption{The four types of codeword $\set{I}$ with even $L$ and $|d^+(\set{I})|<8$.}
\label{table:w8}
\end{table}

For example, let $L=28$.
Then $\set{I}_1=\{0, 2, 4\}$, $\set{I}_2=\{0,7, 14\}$ and $\set{I}_3=\{0,9,18\}$, the equi-difference codewords generated by 2, 7 and 9 respectively, form a code in $\SCAC(28,3)$.
We have $d^*(\set{I}_1) = \{2,4,24,26\}$, $d^*(\set{I}_2) = \{7,14,21\}$ and $d^*(\set{I}_3) = \{9,10,18,19\}$.
Notice that $|d^+(\set{I}_2)|=6$ as the generator $g=7=L/4$, and $|d^+(\set{I}_3)|=6$ as the generator $g=9=(L-1)/3$ (i.e., $\set{I}_3$ is non-dispersive).

\section{Upper Bounds on $M_S(L,3)$} \label{sec:bound}
Following the result of $|d^+(\set{I})|$ for different types of codewords in Section 3, we establish upper bounds on $M_S(L,3)$ under different conditions of $L$.
Since any codeword $\set{I}$ in an SCAC with even length and weight three has $|d^+(\set{I})|\geq 8$ except the four cases listed in Table~\ref{table:w8}, we first discuss $M_S(L,3)$ according to the presence of the four codewords: $\set{I}_{E_1}$, $\set{I}_{E_2}$, $\set{I}_{N_1}$ and $\set{I}_{N_2}$.
In the following lemma, therefore, $L$ is classified according to its remainder after dividing $12$.
Note that we only consider $L \geq 18$ in this section as $M_S(L,3)=1$ if $L<18$ (see \cite{SCAC}).

\begin{lemma}\label{lemma:upperbound_1}
Let $L \geq 18$. Then,
\[
M_S(L,3) \leq 
\begin{cases}
\lfloor(L+4)/8\rfloor & \text{ if } L\equiv 0 \text{ (mod $12$)}, \\
\lfloor(L+2)/8\rfloor & \text{ if } L\equiv 4,6,8 \text{ (mod $12$)}, \\
\lfloor L/8\rfloor & \text{ if } L\equiv 2,10 \text{ (mod $12$)}.
\end{cases}
\]
\end{lemma}
\begin{proof}
Let $\set{C}$ be a code in $\SCAC(L,3)$ with $|\set{C}|=M=M_S(L,3)$.
Assume that the numbers of codewords $\set{I}_{E_1}, \set{I}_{E_2}, \set{I}_{N_1}, \set{I}_{N_2}$ in $\set{C}$ are $e_1, e_2, n_1, n_2$, respectively.

We first consider the case of $L\equiv 0$ (mod $12$).
In this case, $e_1\leq 1$, $e_2\leq 1$ and $n_1=n_2=0$ by Table~\ref{table:w8}.
By Proposition~\ref{prop:property}, we have
\begin{equation}\label{eq:L_lowerbound}
L\geq 2+\sum_{\set{I}\in\set{C}} |d^+(\set{I})|.
\end{equation}
Since $|d^+(\set{I}_{E_1})|=6$, $|d^+(\set{I}_{E_2})|=4$, and $|d^+(\set{I})|\geq 8$ if $\set{I}$ is neither $\set{I}_{E_1}$ nor $\set{I}_{E_2}$, by \eqref{eq:L_lowerbound} we have
\begin{align*}
L &\geq 2+6e_1+4e_2+8\big(M-(e_1+e_2)\big) \\
  &= 2 + 8M - 2e_1 - 4e_2 \\
  &\geq 2 + 8M - 2 -4 = 8M-4.
\end{align*}
Hence $M\leq \lfloor(L+4)/8\rfloor$.

The other five cases can be dealt with in the same way.
Then we complete the proof.
\qed
\end{proof}

%
%
%
%

In the following lemma, we investigate the case of $L\equiv 12$ (mod $24$) in more detail.

\begin{lemma}\label{lemma:upperbound_2}
Let $L \geq 18$.
If $L\equiv 12$ (mod $24$), then
\[
M_S(L,3) \leq (L-4)/8.
\]
\end{lemma}
\begin{proof}
Similar to the setting in the proof of Lemma~\ref{lemma:upperbound_1}, let $\set{C}$ be a code in $\SCAC(L,3)$ with $|\set{C}|=M=M_S(L,3)$ and assume that the numbers of codewords $\set{I}_{E_1}, \set{I}_{E_2}, \set{I}_{N_1}, \set{I}_{N_2}$ in $\set{C}$ are $e_1, e_2, n_1, n_2$, respectively.
The conditions $3|L$ and $4|L$ imply that $e_1\leq 1$, $e_2\leq 1$ and $n_1=n_2=0$.
We aim to show that $L\geq 8M+4$.

Observe that 
\[
d^+(\set{I}_{E_1}) = \Big\{\frac{L}{4},\frac{L}{4}+1,\frac{L}{2},\frac{L}{2}+1,\frac{3L}{4},\frac{3L}{4}+1\Big\}.
\]
Since $L/4$ is odd, $d^+(\set{I}_{E_1})$ has four rough gaps.
Among other possible codewords in $\set{C}$, only non-equi-difference codewords may have rough tubes.
Moreover, if a codeword $\set{I}$ has rough tubes, we have $|d^+(\set{I})|\geq 10$ by Proposition~\ref{prop:tube} and the assumption that $L\equiv 12$ (mod $24$).

Assume that there are $t$ non-equi-difference codewords in $\set{I}$.
If $t \geq 1$, then by \eqref{eq:L_lowerbound} we have
\begin{align*}
L &\geq 2 + \sum_{\set{I}\in\set{C}} |d^+(\set{I})| \\
  &\geq 2 + 6e_1 + 4e_2 + 10t + 8\big( M-(e_1+e_2+t)\big) \\
  &= 2 + 8M - 2e_1 - 4e_2 + 2t \\
  &\geq 2 + 8M - 2 -4 +2 = 8M-2.
\end{align*}
If $t=0$, otherwise, then the four rough gaps in $d^+(\set{I}_{E_1})$ are all solitary.
By plugging $\lambda=4$ into Lemma~\ref{lemma:gap}, we have
\begin{align*}
L &\geq 2 + 4 + \sum_{\set{I}\in\set{C}} |d^+(\set{I})| \\
  &\geq 2 + 4 + 6e_1 + 4e_2 + 8\Big( M-(e_1+e_2)\Big) \\
  &= 6 + 8M - 2e_1 - 4e_2 \\
  &\geq 6 + 8M - 2 -4 = 8M.
\end{align*}
In either case, we obtain $L\geq 8M+4$ due to $L\equiv 12$ (mod $24$).
\qed
\end{proof}

We end this section by collecting the results in Lemma~\ref{lemma:upperbound_1} and \ref{lemma:upperbound_2}.

\begin{theorem} \label{thm:upperbound}
Let $L\geq 18$. Then
\[
M_S(L,3) \leq 
\begin{cases}
  L/8 & \text{if } L\equiv 0 \text{ (mod $8$)}, \\
  (L-4)/8 & \text{if } L\equiv 4 \text{ (mod $8$)}, \\
  (L+2)/8 & \text{if } L\equiv 6 \text{ (mod $24$)}, \\
  (L-2)/8 & \text{if } L\equiv 2,10,18 \text{ (mod $24$)}, \\
  (L-6)/8 & \text{if } L\equiv 14,22 \text{ (mod $24$)}. \\
\end{cases}
\]
\end{theorem}

\section{Optimal SCACs} \label{sec:optimal}
In this section we will show that the upper bounds of $M_S(L,3)$ obtained in Theorem~\ref{thm:upperbound} are indeed tight in several cases.
To construct SCACs attaining these upper bounds, we revisit a construction of SCACs from existing CACs proposed in \cite{SCAC}.

Let $\set{C}=\{\set{I}_1,\set{I}_2,\ldots,\set{I}_M\}$ be a CAC of length $L$ and weight $\omega$.
For $j=1,2,\ldots,M$ define $2\set{I}_j=\{2t:\,t\in\set{I}_j\}$. 
By viewing each $2\set{I}_j$ as an $\omega$-subset of $\mathbb{Z}_{2L}$, it is obvious that $\{1,2L-1\}\cap d^*(2\set{I}_j)=\emptyset$ and $d^+(2\set{I}_j)\cap d^+(2\set{I}_k) = \emptyset$ for all $j \neq k$. Thus, $\{2\set{I}_1,2\set{I}_2,\ldots,2\set{I}_M\}$ forms an SCAC of length $2L$ and weight $\omega$ by Proposition~\ref{prop:property}.
Note that the strategy of doubling all elements in $\set{I}_j$ is equivalent to that of padding an extra zero after each entry when considering $\set{I}_j$ as a binary sequence.


\begin{theorem}[\cite{SCAC}]
If there exists a CAC of $M$ codewords in $\CAC(L,\omega)$, then there exists an SCAC of $M$ codewords in $\SCAC(2L,\omega)$.
\label{thm:construction}
\end{theorem}

By Theorem~\ref{thm:construction}, it is easy to see that $M_S(L,\omega)\geq M(L/2,\omega)$ whenever $L$ is even.
Therefore, we can obtain several optimal SCACs by Theorem~\ref{thm:upperbound}, Theorem~\ref{thm:construction} and some known optimal CACs listed in Section~\ref{sec:pre_known}.

\begin{corollary}\label{cor:SCAC_even1}
Let $L\geq 18$. Then,
\begin{enumerate}
\item $M_S(L,3)=(L-4)/8$ if $L\equiv 4$ (mod $8$);
\item $M_S(L,3)=(L-2)/8$ if $L=2^{2t+1}+2$, $2^{2t}-2^{t+1}+2$ or $2^{2t}+2^{t+1}+2$ for some $t$;
\item $M_S(L,3)=(L+2)/8$ if $L=2^{2^t+1}-2$ for some $t$.
\end{enumerate}
\end{corollary}

A bound of $M_S(L,3)$ for each $L\equiv 0$ (mod 8) is also obtained.
\begin{corollary}\label{cor:SCAC_even2}
Let $L=8t$ for some $t\geq 3$. Then,
\[
\frac{L}{8} \geq M_S(L,3) \geq 
\begin{cases}
    7L/64 & \text{if } t\equiv 0 \text{ (mod $8$)},\\
    (7L+8)/64 & \text{if } t\equiv 1 \text{ (mod $8$)},\\
    (7L-48)/64 & \text{if } t\equiv 2,10 \text{ (mod $24$)},\\
    (7L+24)/64 & \text{if } t\equiv 3 \text{ (mod $24$)},\\
    (7L-32)/64 & \text{if } t\equiv 4,20 \text{ (mod $24$)},\\
    (7L-24)/64 & \text{if } t\equiv 5,13 \text{ (mod $24$)},\\
    (7L-16)/64 & \text{if } t\equiv 6 \text{ (mod $8$)},\\
    (7L-8)/64 & \text{if } t\equiv 7 \text{ (mod $8$)},\\
    (7L-40)/64 & \text{if } t\equiv 11,19 \text{ (mod $24$)},\\
    (7L+32)/64 & \text{if } t\equiv 12 \text{ (mod $24$)},\\
    (7L+16)/64 & \text{if } t\equiv 18 \text{ (mod $24$)},\\
    (7L+40)/64 & \text{if } t\equiv 21 \text{ (mod $24$)}.\\
\end{cases}
\]
\end{corollary}

\smallskip
In what follows we consider CACs of odd length and weight three.
Let $\mathcal{C}\in\CAC(L,3)$ with odd $L$ and $\mathcal{I}$ be one of its codewords.
Since $L$ is odd, we have
\begin{equation}
|d^*(\mathcal{I})|=
\begin{cases}
2 & \text{if } \mathcal{I} \text{ is equi-difference with generator } g=\frac{L}{3}, \\
4 & \text{if } \mathcal{I} \text{ is equi-difference with generator } g\neq \frac{L}{3} ,\\
6 & \text{otherwise.}\\
\end{cases}
\label{eq:12}
\end{equation}

We say a code $\C\in\CAC(L,3)$ has \emph{leave} $\Lambda$ if
$$\mathbb{Z}_L \backslash \bigcup_{\mathcal{I}\in\mathcal{C}}d(\mathcal{I})=\Lambda.$$
If $\Lambda$ is empty, then the code $\mathcal{C}$ is said to be \emph{tight}.
By~\eqref{eq:12}, we have the following.

\begin{proposition}\label{pro:optimal_with_leave}
Let $\C$ be a code in $\CAC^e(L,3)$ having leave $\Lambda$, where $L\geq 3$ is an odd integer.
If $|\Lambda|<4$ and $\{\frac{L}{3},\frac{2L}{3}\}\not\subset\Lambda$, then $\mathcal{C}$ is optimal.
Moreover, $$|\mathcal{C}|=M^e(L,3)=M(L,3).$$
\end{proposition}

Let $L\geq3$ be an odd integer and $G(L)$ be a graph with vertex set $V(G)=\{1,2,\ldots, \frac{L-1}2\}$ and edge set $E(G)$, defined by $(a,b)\in E(G)$ if $b\equiv \pm2a$ (mod $L$).
Then the graph $G(L)$ is a union of disjoint cycles.
Note that a loop is considered as a cycle of length $1$, and a pair of multiedges is considered as a cycle of length $2$.
$G(L)$ is useful in finding the number $M^e(L,3)$.
More precisely, an edge $(a,b)$ in $G(L)$ represents the equi-difference codeword $\{0,a,2a\}$ in a code of length $L$, then the number $M^e(L,3)$ is determined by the size of \emph{maximum matching} in $G(L)$.
Let $N_{odd}(L)$ be the number of odd cycles in $G(L)$.
The following equation was given in \cite{Fu_Lo_Shum_2012}.
\begin{equation}
M^e(L,3) = \frac{(L-1)/2-N_{odd}(L)}{2} + \chi(3|L),
\end{equation}
where $\chi(A) = 1$ or 0 depends on the statement $A$ is true or false.

For an odd integer $n>2$ let $e_n$ be the smallest exponent $e\geq 1$ such that $2^e\equiv 1$ (mod $n$), and let $c_n$ be the smallest exponent $c\geq 1$ such that $2^c\equiv\pm 1$ (mod $n$).
The exponent $e_n$ and $c_n$ are called the \emph{multiplicative order} and the \emph{multiplicative suborder} of $2$ modulo $n$, respectively.

For any odd prime $p$, Fu et al.~\cite{Fu_Lo_Shum_2012} characterize the number $N_{odd}(p)$ in terms of $e_p$ and derive a necessary and sufficient condition for a tight CAC of weight three.

\begin{theorem}[\cite{Fu_Lo_Shum_2012}]\label{thm:O_prime}
Let $p$ be an odd prime. Then,
\[
N_{odd}(p)=
\begin{cases}
\frac{p-1}{2e_p} & \text{ if } p\equiv 7 \text{ (mod 8), or } p\equiv 1 \text{ (mod 8) and } e_p \text{ is odd}, \\
\frac{p-1}{e_p} & \text{ if } p\equiv 3 \text{ (mod 8), or } p\equiv 1 \text{ (mod 8) and } 4|(e_p-2), \\
0 & \text{ if } p\equiv 5 \text{ (mod 8), or } p\equiv 1 \text{ (mod 8) and } 4|e_p.
\end{cases}
\]
\end{theorem}

\begin{theorem}[\cite{Fu_Lo_Shum_2012}]\label{thm:tight_codition}
Let $L=\prod_{i=1}^{m}p_i^{r_i}$ be an odd integer, where $p_1<p_2<\ldots<p_m$ are distinct prime factors and each $r_i\in\mathbb{N}$.
There exists a tight equi-difference code $\set{C}\in\CAC(L,3)$ if and only if one of the following holds:
\begin{enumerate}[(a)]
\item $p_1>3$ and each $p_i$ satisfies the third condition in Theorem \ref{thm:O_prime}; or
\item $p_1=3,r_1=1$, and for $i\geq 2$, $p_i$ satisfies the third condition in Theorem~\ref{thm:O_prime}.
\end{enumerate}
\end{theorem}

In $G(L)$, the \emph{standard cycle}, denoted as $\langle 2\rangle_L$, is the cycle which contains $1$.
Given a cycle $C=(s_1,s_2,\ldots,s_t)$ in $G(L)$ and an integer $a$.
The \emph{modulo product} of $C$ by $a$, denoted by $aC$, is the cycle $(a\cdot s_1,a\cdot s_2,\ldots,a\cdot s_t)$ (mod $L$) in $G(L)$ where each item takes symmetry with respect to $L/2$; and, the \emph{normal product} of $C$ by $a$, denoted by $a\times C$, is the cycle $(a\cdot s_1,a\cdot s_2,\ldots,a\cdot s_t)$ in $G(aL)$.
Two cycles are said to be \emph{congruent}, denoted as $\cong$, if they have the same length and one of them is a modulo or normal product of the other one.
It is easy to see that $C\cong a\times C$.
Besides, it is not difficult to see that every cycle in $G(L)$ can be written as $a\langle 2 \rangle_L$ for some integer $1\leq a<L$.
Some properties of $G(L)$ and $c_L$ are given.

\begin{lemma}[\cite{Fu_Lo_Shum_2012}]\label{lem:cyclelength}
Let $L$ be an odd integer.
\begin{enumerate}[(1)]
\item $c_L$ divides $\varphi(L)/2$.
\item Let $a\langle 2 \rangle_L$ be a cycle in $G(L)$ for some integer $a$.
If $gcd(a,L)=d$, then $a\langle 2 \rangle_L \cong \langle 2 \rangle_{\frac{L}{d}}$.
In particular, $|a\langle 2 \rangle_L| = |\langle 2 \rangle_{\frac{L}{d}}| = c_{\frac{L}{d}}$.
\end{enumerate}
\end{lemma}

We now consider equi-difference CACs with small leave set $\Lambda$.
The main result is as follows.
\begin{theorem}\label{thm:1_leave_condition}
Let $L=\prod_{i=1}^{m}p_i^{r_i}$ be an odd integer, where $p_1<p_2<\ldots<p_m$ are distinct prime factors and each $r_i\in\mathbb{N}$.
There exists an equi-difference code $\set{C}\in\CAC(L,3)$ with leave $\Lambda$ of size $2$, $\Lambda\neq\{\frac{L}{3},\frac{2L}{3}\}$, if one of the followings holds:
\begin{enumerate}[(a)]
\item $p_1>3$ and each $p_i$ satisfies the third condition in Theorem \ref{thm:O_prime} with exactly one exception, say $p_t$, which satisfies $c_{p_t}=\frac{p_t-1}{2}$ and $r_t$=1; or
\item $p_1=3,r_1=1$, and for $i\geq 2$, $p_i$ satisfies the third condition in Theorem~\ref{thm:O_prime} with exactly one exception, say $p_t$, which satisfies $c_{p_t}=e_{p_t}=\frac{p_t-1}{2}$ and $r_t$=1.
\item $p_1=3,r_1=2$, and for $i\geq 2$, $p_i$ satisfies the third condition in Theorem~\ref{thm:O_prime}.
\end{enumerate}
\end{theorem}
\begin{proof}
There exists such a code if and only if (i) $N_{odd}(L)=1$ and $3\nmid L$ or (ii) $N_{odd}(L)=2$ and $3|L$.
In the following we shall prove that conditions (a) implies (i) and conditions (b) and (c) imply (ii).

(a)$\Rightarrow$(i):
Let $k$ be a factor of $L$.
We first claim that $c_k \text{ is odd if and only if } k=p_t$.
It it clear that $c_{p_t}$ is odd.
Assume that $k$ is a multiple of some prime factor $p\neq p_t$.
Since $2^{e_k}\equiv$ (mod $k$) implies $2^{e_k}\equiv$ (mod $p$), we have $e_p|e_k$.
Suppose to the contrary that $c_k$ is odd.
By Lemma~\ref{lem:cyclelength}(2), $\frac{k}{p}\langle 2\rangle_k\cong\langle 2\rangle_p$.
This implies that $c_p$ is odd, which contradicts to $N_{odd}(p)=0$.

Since each cycle in $G(L)$ can be written as the form $a\langle 2\rangle_L$, where $a$ is an integer in its cycle.
Lemma~\ref{lem:cyclelength}(2) says that $a\langle 2\rangle_L\cong\langle 2\rangle_{\frac{L}{d}}$ where $d=\gcd(a,L)$, then the length of $a\langle 2\rangle_L$ is odd only when $a=\frac{L}{p_t}$.
Hence, $N_{odd}=1$.

(b)$\Rightarrow$(ii):
Let $k$ be a factor of $L$.
Similar to above argument, $c_k$ is even if $k$ is a multiple of some prime factor $p\neq 3, p_t$; and, $c_k$ is odd if $k=3$ or $p_t$.
Therefore, it suffices to claim that $c_{3p_t}$ is even.
We shall prove a stronger property that
$$c_{3p_t}=e_{3p_t}=p_t-1.$$
Note that $c_n=\frac{e_n}{2}$ if and only if $2^a\equiv -1$ (mod $n$) for some $a$.
Suppose to the contrary that $c_{3p_t}=\frac{e_{3p_t}}{2}$.
Then $2^a\equiv -1$ (mod $3p_t$) for some $a$.
This implies that $2^a\equiv -1$ (mod $p_t$) and thus $c_{p_t}=e_{p_t}/2$, a contradiction to the original assumption.
So, we have $c_{3p_t}=e_{3p_t}$.
In addition, $e_3|e_{3p_t}$ and $e_{p_t}|e_{3p_t}$ imply that $(p_t-1)|e_{3p_t}$.
By Lemma~\ref{lem:cyclelength}(1), $c_{3p_t}$ divides $\varphi(3p_t)/2$, we have
$$c_{3p_t}=e_{3p_t}=\frac{\varphi(3p_t)}{2}=p_t-1.$$
This completes the second case.

(c)$\Rightarrow$(ii):
Notice that $\frac{L}{3}\langle 2\rangle_L$ and $\frac{L}{9}\langle 2\rangle_L$ are two odd cycles in $G(L)$.
Then the result follows from above arguments.
\qed
\end{proof}

A \emph{safe prime} is a prime number $p$ such that $\frac{p-1}{2}$ is also a prime.
It is easy to see that $c_p=\frac{p-1}{2}$ if $p$ is a safe prime.
Moreover, if $p\equiv 7$ (mod $8$), then $c_p=e_p=\frac{p-1}{2}$ (by the first condition in Theorem~\ref{thm:O_prime}).
The following result is derived from Proposition~\ref{pro:optimal_with_leave} and Theorem~\ref{thm:1_leave_condition}.

\begin{corollary}\label{cor:L_no_3factor}
Let $L>3$ be an odd integer.
Then if $3\nmid L$ we have
\begin{enumerate}
  \item $M(L,3)=M^e(L,3)=\frac{L-1}{4}$ if $p\equiv 5$ (mod $8$) for every prime factor $p$;
  \item $M(L,3)=M^e(L,3)=\frac{L-3}{4}$ if there exists exactly one safe prime factor $\hat{p}$ with $\hat{p}^2\nmid L$ and $p\equiv 5$ (mod $8$) for any other prime factor $p$. \end{enumerate}
If $3|L$ and $9\nmid L$, then we have
\begin{enumerate}
\setcounter{enumi}{2}
\item $M(L,3)=M^e(L,3)=\frac{L+1}{4}$ if $p\equiv 5$ (mod $8$) for every prime factor $p$;
  \item $M(L,3)=M^e(L,3)=\frac{L-1}{4}$ if there exists exactly one safe prime factor $\hat{p}\equiv 7$ (mod $8$) with $\hat{p}^2\nmid L$, and $p\equiv 5$ (mod $8$) for any other prime factor $p>3$.
\end{enumerate}
If $9|L$ and $27\nmid L$, then we have
\begin{enumerate}
\setcounter{enumi}{4}
\item $M(L,3)=M^e(L,3)=\frac{L-1}{4}$ if $p\equiv 5$ (mod $8$) for every prime factor $p>3$.
\end{enumerate}
\end{corollary}

\medskip

{\em Remark:}
Levenshtein and Tonchev~\cite[Theorem 7]{LT05} proved that for odd primes $L$ and $p$, $M(L,3)=\frac{L-1}{4}$ if $L=4p+1$ and $M(L,3)=\frac{L-3}{4}$ if $L=2p+1$.
These two results can be obtained from Corollary~\ref{cor:L_no_3factor} (i) and (ii).

\smallskip
By Theorem~\ref{thm:upperbound}, Theorem~\ref{thm:construction} and Corollary~\ref{cor:L_no_3factor}, we have the following results.

\begin{corollary}\label{cor:SCAC_odd}
Let $L$ be an even integer.
Then we have
\begin{enumerate}
  \item $M_S(L,3)=(L-2)/8$ if $6 \nmid L$ and $L/2$ satisfies the condition of (i) in Corollary~\ref{cor:L_no_3factor};
  \item $M_S(L,3)=(L-6)/8$ if $6 \nmid L$ and $L/2$ satisfies the condition of (ii) in Corollary~\ref{cor:L_no_3factor};
  \item $M_S(L,3)=(L+2)/8$ if $6 | L$, $18 \nmid L$ and $L/2$ satisfies the condition of (iii) in Corollary~\ref{cor:L_no_3factor};
  \item $M_S(L,3)=(L-2)/8$ if $6 | L$, $18 \nmid L$ and $L/2$ satisfies the condition of (iv) in Corollary~\ref{cor:L_no_3factor};
  \item $M_S(L,3)=(L-2)/8$ if $18 | L$, $54 \nmid L$ and $L/2$ satisfies the condition of (v) in Corollary~\ref{cor:L_no_3factor}.
\end{enumerate}
\end{corollary}

\section{Conclusion} \label{sec:conclusion}
We establish in Theorem~\ref{thm:upperbound} upper bounds on the size of SCAC of even length and weight three, which improve previously known upper bounds in~\cite{SCAC}.
The new bounds all increase approximately with slope $1/8$ as a function of length $L$.
By constructing SCACs with some optimal CACs, we show the obtained upper bounds are tight in several cases, as stated in Corollary~\ref{cor:SCAC_even1} and Corollary~\ref{cor:SCAC_odd}.
In addition, some new optimal CACs are given in Theorem~\ref{thm:1_leave_condition}.

\bigskip
\noindent {\bf Acknowledgments} \quad
The authors would like to express their gratitude to the referees for their helpful comments in improving the presentation of this paper.
This work was supported by the Hong Kong RGC Earmaked Grant CUHK414012, the National Natural Science Foundation of China (No. 61301107 and 61174060), the Shenzhen Knowledge Innovation Program JCYJ20130401-172046453 and the Specialized Research Fund for the Doctoral Program of Higher Education of China (No. 20133219120010).


\end{document}